\documentclass[a4paper,conference]{IEEEtran}
\ifCLASSINFOpdf
\else
\fi
\usepackage{amsmath}
\usepackage{graphicx,,epsfig,amsmath,latexsym,amssymb,verbatim,color}
\usepackage{dsfont}
\usepackage{theorem}\usepackage{url}
\newtheorem{definition}{Definition}
\newtheorem{proposition}[definition]{Proposition}
\newtheorem{lemma}[definition]{Lemma}

\newtheorem{theorem}[definition]{Theorem}
\newtheorem{corollary}[definition]{Corollary}

\def\squareforqed{\hbox{\rlap{$\sqcap$}$\sqcup$}}
\def\qed{\ifmmode\squareforqed\else{\unskip\nobreak\hfil
\penalty50\hskip1em\null\nobreak\hfil\squareforqed
\parfillskip=0pt\finalhyphendemerits=0\endgraf}\fi}
\def\endenv{\ifmmode\;\else{\unskip\nobreak\hfil
\penalty50\hskip1em\null\nobreak\hfil\;
\parfillskip=0pt\finalhyphendemerits=0\endgraf}\fi}
\newenvironment{proof}{\noindent \textbf{{Proof~} }}{\qed}

\newcommand{\nc}{\newcommand}
\nc{\rnc}{\renewcommand}
\nc{\beg}{\begin{equation}}
\nc{\eeq}{{\end{equation}}}
\nc{\beqa}{\begin{eqnarray}}
\nc{\eeqa}{\end{eqnarray}}
\nc{\lbar}[1]{\overline{#1}}
\nc{\bra}[1]{\langle#1|}
\nc{\ket}[1]{|#1\rangle}
\nc{\ketbra}[2]{|#1\rangle\!\langle#2|}
\nc{\braket}[2]{\langle#1|#2\rangle}

\nc{\proj}[1]{| #1\rangle\!\langle #1 |}
\nc{\avg}[1]{\langle#1\rangle}
\nc{\Rank}{\operatorname{Rank}}
\nc{\smfrac}[2]{\mbox{$\frac{#1}{#2}$}}
\nc{\tr}{\operatorname{Tr}}
\nc{\ox}{\otimes}
\nc{\dg}{\dagger}
\nc{\dn}{\downarrow}
\nc{\cA}{{\cal A}}
\nc{\cB}{{\cal B}}
\nc{\cC}{{\cal C}}
\nc{\cD}{{\cal D}}
\nc{\cE}{{\cal E}}
\nc{\cF}{{\cal F}}
\nc{\cG}{{\cal G}}
\nc{\cH}{{\cal H}}
\nc{\cI}{{\cal I}}
\nc{\cJ}{{\cal J}}
\nc{\cK}{{\cal K}}
\nc{\cL}{{\cal L}}
\nc{\cM}{{\cal M}}
\nc{\cN}{{\cal N}}
\nc{\cO}{{\cal O}}
\nc{\cP}{{\cal P}}
\nc{\cQ}{{\cal Q}}
\nc{\cR}{{\cal R}}
\nc{\cS}{{\cal S}}
\nc{\cT}{{\cal T}}
\nc{\cX}{{\cal X}}
\nc{\cZ}{{\cal Z}}
\nc{\csupp}{{\operatorname{csupp}}}
\nc{\qsupp}{{\operatorname{qsupp}}}
\nc{\var}{{\operatorname{var}}}
\nc{\rar}{\rightarrow}
\nc{\lrar}{\longrightarrow}
\nc{\polylog}{{\operatorname{polylog}}}
\nc{\1}{{\mathds{1}}}
\nc{\wt}{{\operatorname{wt}}}
\nc{\av}[1]{{\left\langle {#1} \right\rangle}}

\def\S{\Sigma}

\nc{\RR}{{{\mathbb R}}}
\nc{\CC}{{{\mathbb C}}}
\nc{\FF}{{{\mathbb F}}}
\nc{\NN}{{{\mathbb N}}}
\nc{\ZZ}{{{\mathbb Z}}}
\nc{\PP}{{{\mathbb P}}}
\nc{\QQ}{{{\mathbb Q}}}
\nc{\UU}{{{\mathbb U}}}
\nc{\EE}{{{\mathbb E}}}
\nc{\id}{{\operatorname{id}}}

\nc{\CHSH}{{\operatorname{CHSH}}}

\nc{\be}{\begin{equation}}
\nc{\ee}{{\end{equation}}}
\nc{\bea}{\begin{eqnarray}}
\nc{\eea}{\end{eqnarray}}
\nc{\<}{\langle}
\rnc{\>}{\rangle}
\nc{\Hom}[2]{\mbox{Hom}(\CC^{#1},\CC^{#2})}
\nc{\rU}{\mbox{U}}

\nc{\ob}[1]{#1}

\nc{\SEP}{{\text{SEP}}}
\nc{\NS}{{\text{NS}}}
\nc{\LOCC}{{\text{LOCC}}}
\nc{\PPT}{{\text{PPT}}}
\nc{\EXT}{{\text{EXT}}}
\nc{\Sym}{{\operatorname{Sym}}}

\nc{\ERLO}{{E_{\text{r,LO}}}}
\nc{\ERLOCC}{{E_{\text{r,LOCC}}}}
\nc{\ERPPT}{{E_{\text{r,PPT}}}}
\nc{\ERLOCCinfty}{{E^{\infty}_{\text{r,LOCC}}}}
\nc{\Aram}{{\operatorname{\sf A}}}

\begin{document}
%
\title{On the quantum no-signalling assisted \\zero-error classical simulation cost of \\non-commutative bipartite graphs}



%
\author{\IEEEauthorblockN{Xin Wang\IEEEauthorrefmark{1},
Runyao Duan\IEEEauthorrefmark{1}\IEEEauthorrefmark{2}}
\IEEEauthorblockA{\IEEEauthorrefmark{1}Centre for Quantum Computation and Intelligent Systems\\ Faculty of Engineering and Information Technology\\
University of Technology Sydney (UTS),
NSW 2007, Australia}
\IEEEauthorblockA{\IEEEauthorrefmark{2}UTS-AMSS Joint Research Laboratory for Quantum Computation and Quantum Information Processing\\ Academy of Mathematics and Systems Science\\ Chinese Academy of Sciences, Beijing 100190, China}
Email: xin.wang-8@student.uts.edu.au, runyao.duan@uts.edu.au}


\maketitle

\begin{abstract}
Using one channel to simulate another exactly with the aid of quantum no-signalling correlations has been studied recently. The one-shot no-signalling assisted classical zero-error simulation cost of non-commutative bipartite graphs has been formulated as semidefinite programms  [Duan and Winter, IEEE Trans. Inf. Theory 62, 891 (2016)]. Before our work, it was unknown whether the one-shot (or asymptotic) no-signalling assisted zero-error classical simulation cost for general non-commutative graphs is multiplicative (resp. additive) or not.  In this paper we address these issues and give a general sufficient condition for the multiplicativity of the one-shot simulation cost and the additivity of the asymptotic simulation cost of non-commutative bipartite graphs, which include all known cases such as extremal graphs and classical-quantum graphs. Applying this condition, we exhibit a large class of so-called \emph{cheapest-full-rank graphs} whose  asymptotic zero-error simulation cost is given by the one-shot simulation cost. Finally, we disprove the multiplicativity of one-shot simulation cost by explicitly constructing a special class of qubit-qutrit non-commutative bipartite graphs. 
\end{abstract}


%
\IEEEpeerreviewmaketitle

\section{Introduction}
Channel simulation is a fundamental problem in information theory, which concerns how to use a  channel
$\cN$ from Alice (A) to Bob (B) to simulate another channel $\cM$ also from A to B \cite{KretschmannWerner:tema}.
Shannon's celebrated noisy channel coding theorem determines the capability of any noisy channel $\cN$ to simulate an noiseless channel \cite{Shannon1948} and the dual theorem ``reverse Shannon theorem'' was proved recently \cite{BSST2003}. According to different resources available between A and B, this simulation problem has many variants and the case when A and B share unlimited amount of entanglement has been completely solved \cite{BSST2003}. To optimally simulate $\cM$ in the asymptotic setting, the rate is determined by the entanglement-assisted classical capacity of $\cN$ and  $\cM$ \cite{BDHS+2009,QRST-simple}. Furthermore, this rate cannot be improved even with no-signalling correlations or feedback \cite{BDHS+2009}.

In the zero-error setting \cite{Shannon1956} , recently the quantum zero-error information theory has been studied and the problem becomes more complex since many unexpected phenomena were observed such as the super-activation of noisy channels \cite{Duan2009,DS2008,  CCH2009, CS2012} as well as  the assistance of shared entanglement in zero-error communication \cite{CLMW2010, LMMO+2012}.

Quantum no-signalling correlations (QNSC) are introduced as two-input and two-output quantum channels with
the no-signalling constraints. And such correlations have been studied in the relativistic causality of quantum
operations \cite{BGNP2001,ESW2001,PHHH2006,OCB2012}. Cubitt et al. \cite{CLMW2011} first introduced classical no-signalling correlations into the zero-error classical communication problem. They also observed a kind of reversibility between no-signalling assisted zero-error capacity and exact simulation \cite{CLMW2011}. Duan and Winter \cite{Duan2014} further introduced quantum non-signalling correlations into the zero-error communication problem and formulated  both capacity and simulation cost problems as semidefinite programmings (SDPs) \cite{SDP} which depend only on the non-commutative bipartite graph $K$.
To be specific,  QNSC is a bipartite completely positive and trace-preserving linear map $\Pi: \cL(\cA_i)\otimes \cL(\cB_i)\rightarrow \cL(\cA_o)\otimes \cL(\cB_o)$, where the subscripts $i$ and $o$ stand for input and output, respectively.
Let the  Choi-Jamio\l{}kowski matrix of $\Pi$ be
$\Omega_{\cA_i'\cA_o\cB_i'\cB_o}= (\1_{\cA_i'}\otimes \1_{\cB_i'}\otimes \Pi) (\Phi_{\cA_i\cA_i'}\otimes \Phi_{\cB_i\cB_i'})$, where 
$\Phi_{\cA_i\cA_i'}= \ketbra{\Phi_{\cA_i\cA_i'}}{\Phi_{\cA_i\cA_i'}}$, and $\ket{\Phi_{\cA_i\cA_i'}}=\sum_k \ket{k_{\cA_i}}\ket{k_{\cA_i'}}$ is the un-normalized  maximally-entangled state.The following constraints are required for $\Pi$ to be QNSC \cite{Duan2014}:
\begin{align*}
\Omega_{\cA_i'\cA_o\cB_i'\cB_o}\geq 0, \ \tr_{\cA_o\cB_o}{\Omega_{\cA_i'\cA_o\cB_i'\cB_o}}=\1_{\cA_i'\cB_i'},&\\
\tr_{\cA_o\cA_i'}{\Omega_{\cA_i'\cA_o\cB_i'\cB_o}  X^{T}_{\cA_i'}}=0, \forall \tr{X}=0,&\\
\tr_{\cB_o\cB_i'}{\Omega_{\cA_i'\cA_o\cB_i'\cB_o}  Y^{T}_{\cB_i'}}=0, \forall \tr{Y}=0.&
\end{align*}
The new map $\cM^{A_i\rightarrow B_o}=\Pi^{A_i\ox B_i\rightarrow A_o\ox B_o}\circ \cE^{A_o\rightarrow B_i}$ by composing $\cN$ and $\Pi$
can be constructed as illustrated  in Figure \ref{fig:QNSC}. 
\begin{figure}[ht]
\centering
 \includegraphics[width=5.1cm]{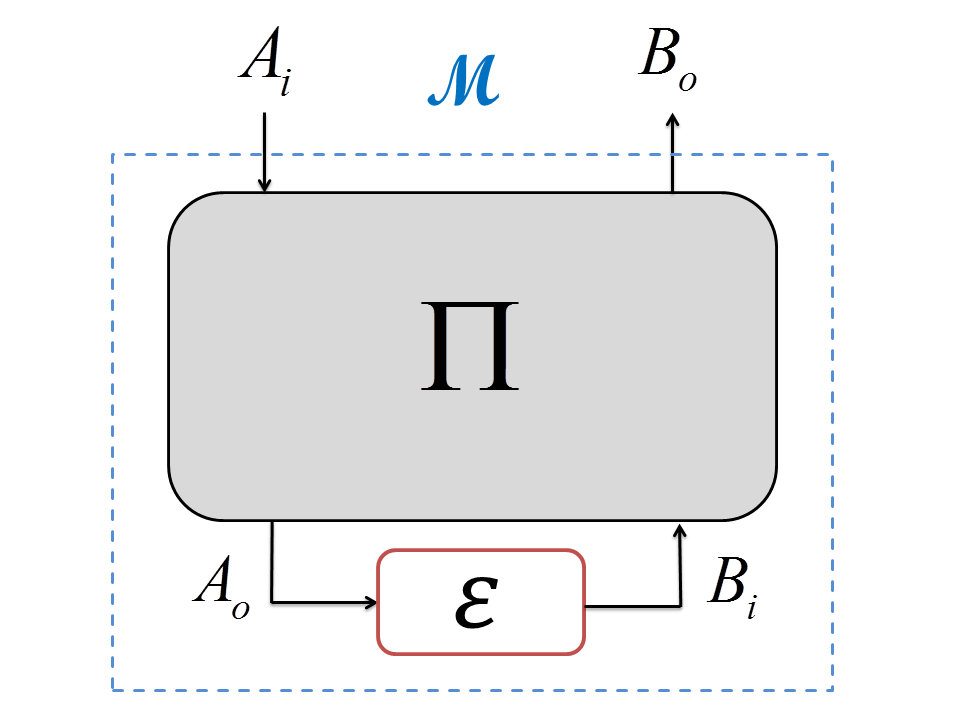}
  \caption{Implementing a channel $\cM$ using another channel $\cE$ with QNSC $\Pi$ between Alice and Bob. }
\label{fig:QNSC}
\end{figure}
The simulation cost  problem concerns how much zero-error communication is required to simulate a noisy channel exactly. Particularly,
the \emph{one-shot} zero-error classical simulation cost of $\cN$ assisted by $\Pi$
is the least noiseless symbols $m$ from $A_o$ to $B_i$ so that $\cM$ can simulate $\cN$.
 In \cite{Duan2014}, the one-shot simulation cost of a quantum channel $\cN$ is given by
\begin{equation}
 \S(\cN)= \min \tr T_B,\ {\rm s.t. }\ J_{AB}\leq \1_A\ox T_B.
\end{equation}
Its dual SDP is
$$\S(\cN) = \max  \tr({J_{AB}}{U_{AB}}) ,\ {\rm s.t.  }\ {{U_{AB}} \ge 0,\tr_A {U_{AB}} = {\1_B}},$$
where $J_{AB}$ is the Choi-Jamio\l{}kowski matrix of $\cN$. By strong duality, the values of both the primal and the dual SDP coincide.
The so-called ``non-commutative graph theory'' was first suggested in \cite{DSW2010} as the non-commutative graph associated with the channel captures the zero-error communication properties, thus playing a similar role to
confusability graph. Let $\cN(\rho)=\sum_k E_k\rho E_k^\dag$ be a quantum channel from $\cL(A)$
to $\cL(B)$, where $\sum_k E_k^\dag E_k=\1_{A}$ and $K=K(\cN)=\operatorname{span}\{E_k\}$ denotes the Choi-Kraus operator space of $\cN$.  The zero-error classical capacity of a quantum channel in the presence of quantum feedback only depends on the Choi-Kraus operator space of the
channel \cite{DSW2015}. That is to say, the Choi-Kraus operator space plays a
role that is quite similar to the bipartite graph. Such Choi-Kraus operator space $K$ is alternatively called ``non-commutative bipartite graph'' since it is clear that any classical channel induces a bipartite graph and a confusability
graph, while a quantum channel induces a non-commutative bipartite graph
together with a non-commutative graph \cite{Duan2014}.

Back to the simulation cost problem, since there might be more than one channel with Choi-Kraus operator space
included in $K$, the exact simulation cost of the ``cheapest'' one among these channels was defined as the one-shot zero-error classical simulation cost of $K$ \cite{Duan2014}: 
$\S(K)=\min \{\S(\cN): \cN \text{ is quantum channel and } K(\cN)<K\}$,
where $K(\cN)<K$ means that $K(\cN)$ is a subspace of $K$.
Then the one-shot zero-error classical simulation cost of a non-commutative bipartite graph $K$ is given by \cite{Duan2014}
\begin{equation}\begin{split}
  \label{eq:Sigma}
  \S(K) &= \min \tr T_B \ \text{ s.t. }\  0 \leq V_{AB} \leq \1_A \ox T_B, \\
        &\phantom{= \min \tr T_B \text{ s.t. }} \tr_B V_{AB} = \1_A, \\
        &\phantom{= \min \tr T_B \text{ s.t. }} \tr (\1-P)_{AB}V_{AB} = 0.
\end{split}\end{equation}
Its dual SDP is
\begin{equation}\begin{split}\label{maxSDP-K}
  \S(K) &= \max \tr S_A \ \text{ s.t. }\  0 \leq U_{AB},\ \tr_A U_{AB} = \1_B, \\
        &\phantom{= \max \tr S_A \text{ }} P_{AB}(S_A \ox \1_B - U_{AB})P_{AB} \leq 0,
\end{split}\end{equation}
where $P_{AB}$ denotes the projection onto the support of the Choi-Jamio\l{}kowski matrix of $\cN$.
Then by strong duality, values of both the primal and the dual SDP coincide. It is evident that $\S(K)$ is sub-multiplicative, which means that for two non-commutative bipartite graphs $K_1$ and $K_2$, $\S({K_1} \otimes {K_2}) \le \S({K_1})\S({K_2})$. Furthermore, the multiplicativity of $\S(K)$ for classical-quantum (cq) graphs as well as extremal graphs were known but the general case was left as an open problem \cite{Duan2014}. By the regularization, the no-signalling assisted zero-error simulation cost is
$$S_{0,NS}(K)  = \inf_{n\geq 1} \frac1n \log \S\left(K^{\ox n}\right).$$

As noted in previous work \cite{Duan2014, DSW2015},
$$C_{0,NS}(K)\le C_{\text{minE}}(K)\le S_{0,NS}(K),$$
where $C_{0,NS}(K)$ is the QSNC assisted classical zero-error capacity and $C_{\text{minE}}(K)$ is the minimum of the entanglement-assisted classical capacity \cite{BSST2003, Bennett1999} of quantum
channels $\cN$ such that $K(\cN) < K$.

Semidefinite programs \cite{SDP} can be solved in polynomial time in the program description \cite{Khachiyan1980} and there exist several different algorithms employing interior point methods  which can compute the optimum value of semidefinite programs efficiently \cite{Alizadeh1995, DeKlerk2002}. The CVX software package \cite{CVX} for MATLAB allows one to solve semidefinite programs efficiently.

In this paper, we focus on the multiplicativity of $\S(K)$ for general non-commutative bipartite graph $K$. We start from the simulation cost of two different graphs and give a sufficient condition which contains all the  known multiplicative cases such as cq graphs and extremal graphs. Then we consider about the simulation cost $\S(K)$ when the ``cheapest'' subspace is full-rank and prove the multiplicativity of one-shot simulation cost in this case. We further explicitly construct a special class of non-commutative bipartite graphs $K_\alpha$  whose one-shot simulation cost is non-multiplicative.  We also exploit some more properties of $K_\alpha$ as well as cheapest-low-rank graphs. Finally, we exhibit a lower bound in order to offer an estimation of the asymptotic simulation cost.

\section{Main results}
\subsection{A sufficient condition of the multiplicativity of simulation cost}
\begin{theorem}\label{sufficient}
Let ${K_1}$ and  ${K_2}$ be non-commutative bipartite graphs
of two quantum channels $\cN_1:\cL(A_1) \to \cL(B_1)$ and ${\cN_2}:\cL(A_2) \to \cL(B_2)$ with
 support projections $P_{{A_1}{B_1}}$ and $P_{{A_2}{B_2}}$, respectively.  Suppose the optimal solutions of SDP(\ref{maxSDP-K}) for $\S({K_1})$ and $\S(K_2)$ are $\left\{ {{S_{{A_1}}},{U_1}} \right\}$ and $\left\{ {{S_{{A_2}}},{U_2}} \right\}$.
If at least one of $S_{{A_1}}$ and $S_{{A_2}}$ satisfy
\begin{equation}\label{sufficient condition}
{P_{{A_i}{B_i}}}({S_{{A_i}}} \otimes {\1_{{B_i}}}){P_{{A_i}{B_i}}} \ge 0, i=1~\text{or}~2,
\end{equation}
then
$$\S({K_1} \otimes {K_2}) = \S({K_1})  \S({K_2}).$$
Furthermore, $$S_{0,NS}(K_1 \ox K_2)=S_{0,NS}(K_1) + S_{0,NS}(K_2).$$
\end{theorem}
\begin{proof}
It is obvious that ${U_1} \otimes {U_2} \ge 0 \text{ and } \tr_{{A_1}{A_2}}({U_1} \otimes {U_2}) = {\1_{{B_1}{B_2}}}$. For convenience, let $P_{{A_1}{B_1}}=P_1$ and $P_{{A_2}{B_2}}=P_2$.
Without loss of generality, we assume that ${P_2}({S_{{A_2}}} \otimes {\1_{{B_2}}}){P_2} \ge 0$.
From the last constraint of SDP(2), we have that
${P_1}\left( {{S_{{A_1}}} \otimes {\1_{{B_1}}}} \right){P_1} \le {P_1}{U_1}{P_1}$
and
${P_2}({S_{{A_2}}} \otimes {\1_{{B_2}}}){P_2} \le {P_2}{U_2}{P_2}$.
Note that
${P_1}\left( {{S_{{A_1}}} \otimes {\1_{{B_1}}}} \right){P_1} \otimes {P_2}({S_{{A_2}}} \otimes {\1_{{B_2}}}){P_2} \le {P_1}{U_1}{P_1} \otimes {P_2}({S_{{A_2}}} \otimes {\1_{{B_2}}}){P_2}$.
It is easy to see that
\begin{equation}
\begin{split}
&{P_1} \otimes {P_2}\left( {{S_{{A_1}}} \otimes {S_{{A_2}}} \otimes {\1_{{B_1}{B_2}}} - {U_1} \otimes {U_2}} \right){P_1} \otimes {P_2}\\
 \le &{P_1}{U_1}{P_1} \otimes [{P_2}({S_{{A_2}}} \otimes {\1_{{B_2}}}){P_2} - {P_2}{U_2}{P_2}] \le 0.
 \end{split}
 \end{equation}
 
Hence, $\left\{ {{S_{{A_1}}} \otimes {S_{{A_2}}}, {U_1} \otimes {U_2}} \right\}$ is a feasible solution of SDP(\ref{maxSDP-K}) for $\S({K_1} \otimes {K_2})$, which means that
$\S({K_1} \otimes {K_2}) \ge \S({K_1})\S({K_2})$.
Since $\S(K)$ is sub-multiplicative, we can conclude that
$\S({K_1} \otimes {K_2}) = \S({K_1})\S({K_2})$.

Furthermore, for $K_2^{\ox n}$, it is easy to see that $\{S_{A_2}^{\ox n}, U_2^{\ox n} \}$ is a feasible solution of SDP(\ref{maxSDP-K}) for $\S(K_2^{\ox n})$ and ${P_2^{\ox n}}({S_{{A_2}}^{\ox n}} \otimes {\1_{{B_2}}}^{\ox n}){P_2^{\ox n}} \ge 0$. Therefore, $\S(K_2^{\ox n})=\S(K_2)^n$ and 
$$\S[(K_1 \ox K_2 )^{\ox n}]  =\S(K_1^{\ox n} \ox K_2^{\ox n} ) =\S(K_1^{\ox n})\S(K_2^{\ox n}).$$
Hence, 
\begin{align*}
S_{0,NS}(K_1 \ox K_2)&= \inf_{n\geq 1} \frac1n \log \S\left(K_1^{\ox n} \ox K_2^{\ox n}\right)\\
&=\inf_{n\geq 1} \frac1n \log\S(K_1^{\ox n})\S(K_2^{\ox n})\\
&=S_{0,NS}(K_1) + S_{0,NS}(K_2).
\end{align*}
\end{proof}

In \cite{DW2015}, the activated zero-error no-signalling assisted capacity has been studied. Here, we consider about the corresponding simulation cost problem. 
\begin{corollary}
For any non-commutative bipartite graph K, let $\Delta_\ell=\sum_{k=1}^\ell \ketbra{kk}{kk}$ be the non-commutative bipartite graph of a noiseless channel with $\ell$ symbols, then $${\S(K \ox \Delta_\ell)}={\ell}\S(K),$$
which means that noiseless channel cannot reduce the simulation cost of any other  non-commutative bipartite graph.
\end{corollary}
\begin{proof}
It is evident that $\Delta_\ell$ satisfies the condition in Theorem \ref{sufficient}. Then, ${\S(K \ox \Delta_\ell)}=\ell\S(K)$.
\end{proof}

\subsection{Simulation cost of the cheapest-full-rank non-commutative bipartite graph}
\begin{definition}
Given a non-commutative bipartite graph $K$ with support projection $P_{AB}$. Assume the ``cheapest channel'' in this space is $\cN_c$ with Choi-Jamio\l{}kowski matrix $J_{\cN_c}$. $K$ is said to be \textbf{cheapest-full-rank} if there exists $\cN_c$ such that $rank(J_{\cN_c})=rank(P_{AB})$.  Otherwise, $K$ is said to be \textbf{cheapest-low-rank}.
\end{definition}

\begin{lemma}\label{lemma of PAP}
For a quantum channel $\cN$ with Choi-Jamio\l{}kowski matrix $J_{AB}$ and support projection $P_{AB}$,
if $P_{AB}CP_{AB}=P_{AB}DP_{AB}$,
then
$\tr (CJ_{AB})=\tr (DJ_{AB})$.
\end{lemma}
\begin{proof}
It is easy to see that 
\begin{align*}
\tr (CJ_{AB}) &=\tr(CP_{AB}J_{AB}P_{AB})=\tr(P_{AB}CP_{AB}J_{AB})\\
&=\tr(P_{AB}DP_{AB}J_{AB}) =\tr (DJ_{AB}).
\end{align*}
\end{proof}

\begin{proposition}\label{lemma JW}
For any non-commutative bipartite graph K with support projection $P_{AB}$, suppose that the cheapest channel is $\cN_c$ and the optimal solution of SDP (\ref{maxSDP-K}) is $\{S_A, U_{AB}\}$. Assume that
\begin{equation}
P_{AB}(S_A \ox \1_B - U_{AB})P_{AB} = -W_{AB} \text{ and } W_{AB} \ge 0.
\end{equation}
Then, we have that
\begin{equation}\label{lemma2}
    \tr W_{AB}J_{AB}=0,
\end{equation} and  $U_{AB}$  is also the optimal solution of $\S(\cN_c)$, 
 where  $J_{AB}$ is the Choi-Jamio\l{}kowski matrix of  $\cN_c$.
\end{proposition}
\begin{proof}
On one hand, since $\cN_c$ is the cheapest channel, $\S(K)$ will equal to $\S(\cN_c)$, also noting that $\{S_A, U_{AB}\}$ is the optimal solution, we have
\begin{equation}\begin{split}\label{eq1-lemma2}
 \tr{S_A} &=\S(K) =\S(\cN_c) \\
  &= \max  \tr {J_{AB}}{V_{AB}} ,\ {\rm s.t.  }\ {{V_{AB}} \ge 0,\tr_A {V_{AB}} = {\1_B}},\\
   &\ge  \tr {J_{AB}}{U_{AB}}.
\end{split}\end{equation}

On the other hand, it is evident that $W_{AB}= P_{AB}WP_{AB} $, then
$P_{AB}U_{AB}P _{AB} = P _{AB} (W_{AB}+ S_A \ox \1_B)P _{AB}$.
From Lemma \ref{lemma of PAP}, we can conclude that
$\tr U_{AB}J_{AB} = \tr (W_{AB}+ S_A \ox \1_B)J_{AB}=\tr W_{AB}J_{AB}+\tr(S_A \ox \1_B)J_{AB}$.

For Choi-Jamio\l{}kowski matrix $J_{AB}$, we have that
\begin{equation}\begin{split}
\tr({S_A} \otimes {\1_B}){J_{AB}} &= \tr_A\tr_B[({S_A} \otimes {\1_B}){J_{AB}}] \\
&=\tr_A[{S_A}(\tr_B{J_{AB}}) ]= \tr{S_A},
\end{split}\end{equation}
then
\begin{equation}\label{eq2-lemma2}
\tr U_{AB}J_{AB}= \tr W_{AB}J_{AB}+\tr{S_A}.
\end{equation}
Combining (\ref {eq1-lemma2}) and (\ref{eq2-lemma2}), and noting that $W_{AB}, J_{AB} \ge 0$, we can conclude that
$\tr W_{AB}J_{AB}=0$ and  $U_{AB}$  is also the optimal solution of  $\S(\cN_c)$.
\end{proof}

\begin{theorem}\label{full-rank-S}
For any cheapest-full-rank non-commutative bipartite graph $K$, we have
\begin{equation}\begin{split}
  \label{eq:Sigma-dual-hat}
\S(K) &= \max \tr S_A \ \text{ s.t. }\  0 \leq U_{AB},\ \tr_A U_{AB} = \1_B, \\
        &\phantom{= \max \tr S_A \text{ s.t. }} P_{AB}(S_A \ox \1_B - U_{AB})P_{AB} = 0.
\end{split}\end{equation}
Also,
$\S({K} \otimes {K}) = \S({K})  \S({K})$. Consequently, $S_{0,NS}(K)= \log \S(K)$. 

And for any other non-commutative bipartite graph $K'$,
$S_{0,NS}(K\ox K')=S_{0,\NS}(K)+S_{0,NS}(K')$.
\end{theorem}

\begin{proof}
We first assume that $W \ne 0$.
Notice $rank(J_{AB})=rank(P_{AB})$, it is easy to see that $\tr WJ_{AB}>0$,
which contradicts Eq. (\ref{lemma2}).
Hence the assumption is false, and we  can conclude that $P_{AB}(S_A \ox \1_B - U_{AB})P_{AB} = 0$.

Then by Theorem \ref{sufficient}, it is easy to see that $\S(K \ox K) = \S(K)\S(K)$.
Therefore,
$$S_{0,\NS}(K)    = \inf_{n\geq 1} \frac1n \log \S(K^{\ox n})= \log \S(K).$$
Furthermore, for any other  non-commutative bipartite graph $K'$, $S_{0,NS}(K\ox K')= S_{0,\NS}(K)+S_{0,NS}(K')$.
\end{proof}

Noting that any rank-2 Choi-Kraus operator space is always cheapest-full-rank, we have the following immediate corollary.
\begin{corollary}
For any rank-2 Choi-Kraus operator space $K$,
$S_{0,NS}(K)= \log \S(K)$.
And for any other non-commutative bipartite graph $K'$,
$S_{0,NS}(K\ox K')= S_{0,NS}(K)+S_{0,NS}(K')$.
\end{corollary}

\subsection{The one-shot simulation cost is not multiplicative}
We will focus on the non-commutative bipartite graph $K_\alpha$  with
support projection $P_{AB}=\sum \limits_{j=0}^{2} \ketbra {\psi_j} {\psi_j}$, where
$\ket {\psi_0} =\frac{1}{\sqrt 3}(\ket {00}+\ket {01}+\ket {12}), 
\ket {\psi_1} =\cos \alpha \ket {02}+\sin \alpha \ket {11}, 
\ket {\psi_2} =\ket {10}$.

To prove that  $K_\alpha$ ($0<\cos^2 \alpha<1$) is feasible to be a class of feasible non-commutative bipartite graphs, we only need to find a channel $\cN$ with
Choi-Jamio\l{}kowski matrix $J_{AB}$ such that $P_{AB}J_{AB}=J_{AB}$ and $\text{rank}(P_{AB})=\text{rank}(J_{AB})$ .
 Assume that $J_{AB}=\sum \limits_{j=0}^{2} a_j\ketbra {\psi_j} {\psi_j}$,
then it is equivalent to prove that $\tr_B J_{AB}=\1_A$ and $J_{AB} \ge 0$ has a feasible solution. Therefore,
$$
\frac{2}{3}a_0+\cos^2\alpha a_1=1,
a_0+a_1+a_2=2,
a_0, a_1, a_2 >0.
$$
Noting that when we choose $0< a_1<\frac{1}{2}$, $a_0=\frac{3}{2}(1-\cos^2\alpha a_1)$ and $a_2=\frac{1-(2-3\cos^2\alpha )a_1}{2}$ will be positive, which means that  there exists such $J_{AB}$.
Hence, $K_\alpha$ is a feasible noncommutative bipartite graph.

\begin{theorem}
There exists non-commutative bipartite graph K such that $\S(K\ox K)<\S(K)^2$.
\end{theorem}
\begin{proof}
As we have shown above, it is reasonable to focus on $K_\alpha$. 
Then, by semidefinite programming assisted with useful tools CVX \cite{CVX} and QETLAB \cite{QETLAB}, the gap between one-shot and two-shot average no-signalling assisted zero-error simulation cost of $K_\alpha (0.25 \leq \cos^2 \alpha \leq 0.35)$ is presented in  Figure \ref{fig:sc eg}.

To be specfic, when $\alpha=\pi/3$,  it is clear that $\cos^2 \alpha=1/4$ and 
$\ket {\psi_1} =\frac{1}{2} \ket {02}+\frac{\sqrt 3}{2} \ket {11}$.
Assume that $S=3.1102\proj{0}-0.5386\proj{1}$
and $U=\frac{99}{50}\proj{u_1}+\frac{51}{50}\proj{u_2}$, where $\ket {u_1}=\frac{10}{3\sqrt {33}}\ket{00}+\frac{5}{3}\sqrt\frac{2}{33}\ket {01}+\frac{7}{3\sqrt {11}}\ket{12}$ and
$\ket {u_2}=\frac{1}{\sqrt {51}}\ket{02}-\frac{5}{3}\sqrt\frac{2}{17}\ket{10}+\frac{10}{3\sqrt {17}}\ket{11}$, and it can be checked that $U\ge 0$, $\tr_A U=\1_B$ and $P_{AB}(S_A \ox \1_B - U_{AB})P_{AB} \le 0$. Then $\{S, U\}$ is a feasible solution of SDP (\ref{maxSDP-K}) for $\S(K_{\pi/3})$, which means that $\S(K_{\pi/3})\ge \tr S= 2.5716$.
Similarily, we can find a  feasible solution of SDP (\ref{eq:Sigma}) for $\S(K_{\pi/3}\ox K_{\pi/3})$ through Matlab such that $\S(K_{\pi/3}\ox K_{\pi/3})^{1/2}\le 2.57$. (The code is available at \cite{sc code}.) Hence, there is a non-vanishing gap between $\S(K_{\pi/3})$ and $\S(K_{\pi/3}\ox K_{\pi/3})^{1/2}$. 
\end{proof}

\begin{figure}[htbp]
  \centering
  \includegraphics[width=0.41\textwidth]{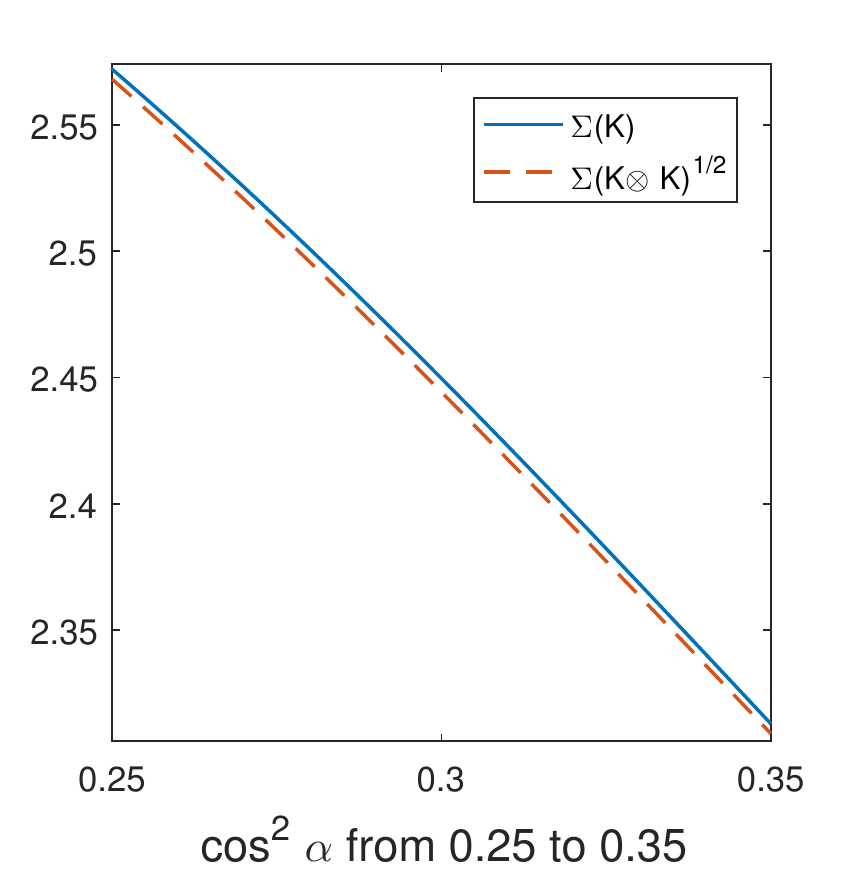}
  \caption{The one-shot (red)  and two-shot average (blue) no-signalling assisted zero-error simulation cost of $K_\alpha$ over the parameter $\alpha$.}\label{fig:sc eg}
\end{figure}

We have shown that one-shot simulation cost of cheapest-full-rank non-commutative bipartite graphs is multiplicative while there are counterexamples for cheapest-low-rank ones.
However, not all cheapest-low-rank graphs have non-multiplicative simulation cost.
Here is one trivial counterexample. Let $K=\text{span}\{\ket 0 \bra 0 ,\ket1 \bra 0 ,\ket 1 \bra 1\}$, the cheapest channel is a constant channel $\cN$ with $E_0=\ket 1 \bra 0$ and $E_1=\ket 1 \bra 1$. In this case, $\S(K\otimes K) = \S(K)\S(K)=1.$
Actually, the simulation cost problem of cheapest-low-rank non-commutative bipartite graphs is complex since it is hard to determine the cheapest subspace under tensor powers.
Therefore, it is difficult to calculate the asymptotic simulation cost of non-multiplicative cases.

In \cite{DSW2015},  $K$ is called non-trivial if there is no constant channel $\cN_0: \rho \to \proj \beta$ with $K(\cN_0) < K$, where  $\ket \beta$ is a state vector. It was known that $K$ is non-trivial if and only if the no-signalling assisted zero-error capacity is positive, say $C_{0,NS}(K)>0$.
Clearly we have the following result.
\begin{proposition}
For any non-commutative bipartite graph $K$, $S_{0,NS}(K)>0$ if and only if $K$ is non-trivial.
\end{proposition}
\begin{proof}
If $K$ is non-trivial, it is obvious that $S_{0,NS}(K)\ge C_{0,NS}(K)>0$. Otherwise, $0\le S_{0,NS}(K)\le S_{0,NS}(\cN_0)=0$, which means that $S_{0,NS}(K)=0$.
\end{proof}

\subsection{A lower bound}
Let us introduce a revised SDP which has the same simplified form in cq-channel case:
\begin{equation}\begin{split}
  \label{eq:Sigma2 dual}
  \S^-(K) &= \max \tr S_A \ \text{ s.t. }\    S_A\ge 0, U_{AB}\ge 0\ \tr_A U_{AB} = \1_B, \\
        &\phantom{= \max \tr S_A \text{ s.t. }} P_{AB}(S_A \ox \1_B - U_{AB})P_{AB} \leq 0,
\end{split}\end{equation}

\begin{lemma}\label{lowerbound}
For any non-commutative bipartite graphs $K_1$ and $K_2$,
$$ \S^-(K_1 \ox K_2)  \ge \S^-(K_1)\S^-(K_2).  $$
Consequently,
$\S^-({K_1})  \S^-({K_2}) \le \S({K_1} \otimes {K_2}) \le \S({K_1})  \S({K_2})$.
\end{lemma}
\begin{proof}
From SDP (\ref {eq:Sigma2 dual}), noting that $P_{AB}(S_A \ox \1_B)P_{AB} \ge 0$, it is easy to prove  $\S^-(K_1 \ox K_2)  \ge \S^-(K_1)\S^-(K_2)$ by similar technique applied in Theorem 3.
Therefore,
$\S^-({K_1})  \S^-({K_2}) \le \S^-({K_1} \otimes {K_2}) \le \S({K_1} \otimes {K_2})\le \S({K_1})  \S({K_2})$.
\end{proof}

\begin{proposition}
For  a general non-commutative bipartite graph $K$,
$$\log {\S^ - }(K) \le {S_{0,NS}}(K) \le \log \S(K).$$
\end{proposition}
\begin{proof}
By Lemma \ref{lowerbound}, it is easy to see that
${\S^ - }{(K)^n} \le \S(K^{ \otimes n}) \le \S{(K)^n}$.
Then, $\log {\S^ - }(K) \le {S_{0,NS}}(K) \le \log \S(K)$.
Also, it is obvious that ${S_{0,NS}}(K)$ will equal to $\log \S(K)$ when $\S^-(K)=\S(K).$
\end{proof}

\section{Conclusions}
In sum, for two different non-commutative bipartite graphs, we give sufficient conditions for the multiplicativity of one-shot simulation cost as well as the additivity of the asymptotic simulation cost. The case of cheapest-full-rank non-commutative bipartite graphs has been completely solved while the cheapest-low-rank graphs have a more complex structure. We further show that the one-shot no-signalling assisted classical zero-error simulation cost of non-commutative bipartite graphs is not multiplicative. We provide a lower bound of $\S(K)$ such that the asymptotic zero-error simulation cost can be estimated by $\log\S^-(K)\le S_{0,NS}(K) \le \log\S(K)$.

It is of great interest to know whether the sufficient condition of multiplicativity in Theorem \ref{sufficient} is also necessary. It also remains unknown about the additivity of the asymptotic simulation cost of general non-commutative bipartite graphs and whether it equals to  $\log\S^-(K)$.

\section*{Acknowledgments}
We would like to thank Andreas Winter for his interest on this topic and for many insightful suggestions. XW would like to thank Ching-Yi Lai for helpful discussions on SDP. This work was partly supported by the Australian Research Council (Grant No. DP120103776 and No. FT120100449) and the National Natural Science Foundation of China (Grant No. 61179030).



\bibliographystyle{IEEEtran}
%

\end{document}